\documentclass[oneeqnum,onefignum]{siamltex}

\usepackage{epsfig}

\usepackage{amssymb}
\usepackage{amsxtra}
\usepackage{amsmath}
\usepackage{latexsym}
\usepackage{graphics}
\usepackage{verbatim}
\usepackage{epsfig}
\usepackage{setspace}

\newtheorem{claim}[theorem]{Claim}
\newtheorem{Def}{Definition}[section]

\def\a{\alpha}

\def\l{\lambda}

\def\Z{{\mathbb Z}}

\def\F{{\mathbb F}}

\def\x{{\mathbf x}}

\def\0{{\mathbf 0}}
\def\1{{\mathbf 1}}

\def\cC{{\mathcal C}}
\def\cD{{\mathcal D}}

\def\cF{{\mathcal F}}
\def\cG{{\mathcal G}}

\def\cM{{\mathcal M}}

\def\cP{{\mathcal P}}

\def\cV{{\mathcal V}}

\def\tcV{\widetilde{\cV}}
\def\tV{\widetilde{V}}
\def\oV{\overline{V}}
\def\oE{\overline{E}}

\def\ocV{\overline{\cV}}
\def\ocG{{\overline{\cG}}}
\def\mfC{{\mathfrak C}}

\def\pw{\text{pw}}
\def\tw{\text{tw}}
\def\cl{\text{cl}}

\def\punc{\setminus\!}
\def\del{\punc}
\def\shorten{/}
\def\sh{\shorten}
\def\con{\shorten}

\def\define{\stackrel{\text{\footnotesize def}}{=}}

\def\disj{\stackrel{\cdot}{\cup}}

\title{Matroid Pathwidth and \\ Code Trellis Complexity\thanks{This 
work was supported in part by a research grant from the Natural Sciences and 
Engineering Research Council (NSERC) of Canada.}}
\author{Navin Kashyap\footnotemark[2]}

\begin{document}

\maketitle

\renewcommand{\thefootnote}{\fnsymbol{footnote}}
\footnotetext[2]{Dept.\ of Mathematics and Statistics, Queen's University,
Kingston, ON, K7L 3N6, Canada. Email: {\tt nkashyap@mast.queensu.ca}}

\renewcommand{\thefootnote}{\arabic{footnote}}
\setcounter{footnote}{0}

\begin{abstract}
We relate the notion of matroid pathwidth to the minimum trellis
state-complexity (which we term trellis-width) of a linear code, 
and to the pathwidth of a graph. By reducing from the problem 
of computing the pathwidth of a graph, we show that the problem 
of determining the pathwidth of a representable matroid is NP-hard. 
Consequently, the problem of computing the trellis-width
of a linear code is also NP-hard. For a finite field $\F$,
we also consider the class of $\F$-representable matroids of 
pathwidth at most $w$, and correspondingly, the family of linear 
codes over $\F$ with trellis-width at most $w$. 
These are easily seen to be minor-closed. 
Since these matroids (and codes) have 
branchwidth at most $w$, a result of Geelen and Whittle shows 
that such matroids (and the corresponding codes) are characterized 
by finitely many excluded minors. We provide the complete list of 
excluded minors for $w=1$, and give a partial list for $w=2$. 
% We also give a polynomial-time algorithm for determining whether 
% or not a given binary matroid (resp.\ code) has pathwidth 
% (resp.\ minimum trellis state-complexity) at most 1.
\end{abstract}

\begin{keywords}
Matroids, pathwidth, linear codes, trellis complexity, NP-hard.
\end{keywords}

\begin{AMS}
05B35, 94B05
\end{AMS}

\pagestyle{myheadings}
\thispagestyle{plain}
\markboth{NAVIN KASHYAP}
{MATROID PATHWIDTH AND CODE TRELLIS COMPLEXITY}

\section{Introduction\label{intro_section}}

The notion of pathwidth of a matroid has received some recent
attention in the matroid theory literature \cite{GGW06}, \cite{hall07}.
This notion has long been studied in the coding theory literature,
where it is used as a measure of trellis complexity of a linear 
code \cite{muder}, \cite{For94}, \cite{vardy}. However, there appears
to be no standard coding-theoretic nomenclature for this notion.
It has been called the state complexity of a code in \cite{horn},
but the use of this term there conflicts slightly with 
its use in \cite{vardy}. So to avoid ambiguity, we will give 
it a new name here --- \emph{trellis-width} --- which acknowledges its 
roots in trellis complexity.

The relationship between matroid pathwidth and code trellis-width
can be made precise as follows. To an arbitrary linear
code $\cC$ over a finite field $\F$, we associate a matroid, 
$M(\cC)$, which is simply the vector matroid, over $\F$,
of any generator matrix of the code. Recall that in coding theory,
a matrix $G$ is called a generator matrix of a code $\cC$,
if $\cC$ is the rowspace of $G$. Consequently, the matroid $M(\cC)$ does 
not depend on the actual choice of the generator matrix, and so
is a characteristic of the code $\cC$. The code $\cC$ may 
in fact be viewed as a representation over $\F$ of the matroid 
$M(\cC)$. The trellis-width of $\cC$ is simply the pathwidth of 
$M(\cC)$; we will give the precise definition of matroid pathwidth 
in Section~\ref{pw_section}.

It has repeatedly been conjectured in the coding theory literature 
that computing the trellis-width of a linear code over a fixed finite
field $\F$ is NP-hard \cite{horn}, \cite{jain}, \cite[Section~5]{vardy}. 
This would imply that the corresponding decision problem 
(over a fixed finite field $\F$) --- given a generator matrix for 
a code $\cC$ over $\F$, and a positive integer $w$, deciding 
whether or not the trellis-width of $\cC$ is at most $w$ ---
is NP-complete. This decision problem has been given various
names --- ``Maximum Partition Rank Permutation'' \cite{horn},
``Maximum Width'' \cite{jain} and ``Trellis State-Complexity'' 
\cite{vardy}.

An equivalent statement of the trellis-width conjecture above
is the following: given a matrix $A$ over $\F$, the problem of 
computing the pathwidth the vector matroid $M[A]$ is NP-hard. 
In this paper, we prove the 
above statement for any fixed field $\F$, not necessarily finite.
Our proof is by reduction from the problem of computing the
pathwidth of a graph, which is known to be NP-hard \cite{arnborg}, 
\cite{bod93}. Thus, in particular, computing the trellis-width 
of a linear code over $\F$ is NP-hard, which settles the aforementioned 
coding-theoretic conjecture.

The situation is rather different if we weaken the trellis-width
decision problem above by \emph{not} considering the integer $w$ 
to be a part of the input to the problem. 
In other words, for a fixed finite field $\F$,
and a \emph{fixed} integer $w > 0$, consider the following problem:
\\[-6pt]

\begin{quote}
given a length-$n$ linear code $\cC$ over $\F$, decide whether or
not $\cC$ has trellis-width at most $w$. \\[-6pt]
\end{quote}
% 
% We give this problem the name Weak Trellis-Width (WTW). 
The equivalent decision problem for matroid pathwidth would be to decide 
(for a fixed finite field $\F$ and integer $w > 0$) whether or not a 
given $\F$-representable matroid has pathwidth at most $w$. 
Based on results from the structure theory of matroids \cite{GGW}, we 
strongly believe that these problems are solvable in polynomial time.

In the process of studying matroids of bounded pathwidth, we 
observe that for any finite field $\F_q = GF(q)$ and integer $w > 0$, 
the class, $\cP_{w,q}$, of $\F_q$-representable matroids
having pathwidth at most $w$, is minor-closed and has finitely many 
excluded minors. As a relatively easy exercise, we show that the 
list of excluded minors for $\cP_{1,q}$ consists of\footnote{In this paper,
we take the connectivity function of a matroid $M$ with ground set $E$ 
and rank function $r$ to be $\l_M(X) = r(X) + r(E - X) - r(E)$
for $X \subset E$. Therefore, what we consider to be matroids of 
pathwidth one would be matroids of pathwidth two in \cite{GGW06}, 
\cite{hall07}.} 
$U_{2,4}$, $M(K_4)$, $M(K_{2,3})$ and $M^*(K_{2,3})$. 
Unfortunately, the problem of finding excluded-minor characterizations 
of $\cP_{w,q}$ for $w > 1$ becomes difficult very quickly. 
We give a list of excluded minors for $\cP_{2,q}$, 
which is probably not complete. 

The rest of the paper is organized as follows. In 
Section~\ref{prelim_section}, we lay down the definitions and notation 
used in the paper.
% and also provide the vocabulary necessary to
% translate matroid-theoretic results into the language of coding theory.
In Section~\ref{NP_section}, we prove that, for any fixed field $\F$,
the problem of computing the pathwidth of an $\F$-representable matroid 
is NP-hard, and therefore, so is the problem of computing the
trellis-width of a linear code over $\F$.
Finally, in Section~\ref{bounded_pw_section}, we consider
the class of matroids $\cP_{w,q}$. We give the complete lists of
excluded minors for $\cP_{1,q}$ and the corresponding family
of linear codes over $\F_q$ having trellis-width at most one. We
also give a partial list of excluded minors for  $\cP_{2,q}$.

\section{Preliminaries\label{prelim_section}}

We assume familiarity with the basic definitions and
notation of matroid theory, as expounded by Oxley \cite{oxley}. 
The main results and proofs in this paper will be given
in the language of matroid theory, rather than that of coding
theory, as it is easier to do so. 
However, as our results may be of some interest to coding theorists,
we make an effort in this section to provide the vocabulary necessary 
to translate the language of matroid theory into that of coding theory. 
Definitions of coding-theoretic terms not explicitly defined here 
can be found in any text on coding theory (\emph{e.g.}, \cite{sloane}).

\subsection{Codes and their Associated Matroids\label{codes_matroids_section}}
Let $\cC$ be a linear code of length $n$ over the finite field $\F_q
= GF(q)$. The dimension of $\cC$ is denoted by $\dim(\cC)$, and
the coordinates of $\cC$ are indexed by the integers from the
set $[n] = \{1,2,\ldots,n\}$ as usual.
We will also associate with the coordinates of $\cC$ a set, $E(\cC)$,
of \emph{coordinate labels}, so that there is a bijection 
$\alpha_\cC: [n] \rightarrow E(\cC)$. The \emph{label sequence}
of $\cC$ is defined to be the $n$-tuple $(\a_1, \a_2, \ldots, \a_n)$, 
where $\a_i = \alpha_\cC(i)$. For notational convenience, we 
will simply let $\alpha_\cC$ denote the label sequence of $\cC$.
Unless specified otherwise (as in the case of code minors
and duals below), we will, by default, set $E(\cC)$ to be $[n]$, 
and $\alpha_\cC$ to be the $n$-tuple $(1,2,3,\ldots,n)$. In such a
case, the label of each coordinate is the same as its index.

Given a code $\cC$ over $\F_q$, specified by a generator matrix $G$,
we define its \emph{associated matroid} $M(\cC)$ to be the vector
matroid, $M[G]$, of $G$. We identify the ground set of $M(\cC)$ with 
$E(\cC)$. Note that if $G$ and $G'$ are distinct generator matrices of
the code $\cC$, then $M[G] = M[G']$, and hence, $M(\cC)$ is independent
of the choice of generator matrix. Thus, any generator matrix
of $\cC$ is an $\F_q$-representation of $M(\cC)$.
% in particular, $\cC$ itself is an $\F_q$-representation of $M(\cC)$. 

Conversely, if $M$ is an $\F_q$-representable matroid, 
and $G$ is an $\F_q$-representation of $M$, then
$M = M(\cC)$ for the code $\cC$ generated by $G$. Thus,
each $\F_q$-representable matroid is associated with some code $\cC$
over $\F_q$.

For any code $\cC$, the dual code, $\cC^\perp$, 
is specified to have the same label sequence as $\cC$, \emph{i.e.},
$\alpha_{\cC^\perp} = \alpha_\cC$. It is a particularly nice
fact \cite[Theorem~2.2.8]{oxley} that the matroids 
associated with $\cC$ and $\cC^\perp$
are dual to each other, \emph{i.e.}, $M(\cC^\perp) = 
(M(\cC))^* \define M^*(\cC)$. 

Given a $J \subset E(\cC)$, we will denote by $\cC \punc J$ 
(resp.\ $\cC \sh J$) the code obtained from $\cC$ by puncturing 
(resp.\ shortening at) those coordinates having labels in $J$. 
Thus, $\cC \sh J = (\cC^\perp \punc J)^\perp$. 
A \emph{minor} of $\cC$ is a code of the form $\cC / X \punc Y$ for 
disjoint subsets $X,Y \subset E(\cC)$. A minor of $\cC$
that is not $\cC$ itself is called a \emph{proper minor} of $\cC$.
The coordinates of a minor of $\cC$ retain their labels from $E(\cC)$. 
More precisely, we set $E(\cC / X \punc Y) = E(\cC) - (X \cup Y)$, and take 
the label sequence of $\cC / X \punc Y$ to be the $(n-|X \cup Y|)$-tuple 
obtained from $\alpha_\cC = (\a_1,\a_2,\ldots,\a_n)$ 
by simply removing those entries that are in $X \cup Y$. 
The operations of puncturing and shortening
correspond to the matroid-theoretic operations of deletion and
contraction, respectively: for $J \subset E(\cC)$,
$$
M(\cC \punc J) = M(\cC) \del J \ \ \ \text{and} \ \ \
M(\cC \sh J) = M(\cC) \con J.
$$

We will find it convenient to use $\cC|_J$ to denote the restriction 
of $\cC$ to the coordinates with labels in $J$, \emph{i.e.}, 
$\cC|_J = \cC \punc J^c$, where $J^c$ denotes the set difference
$E(\cC) - J$. This allows us to express the rank function, 
$r:\ E(\cC) \rightarrow \Z$, of the matroid $M(\cC)$ as follows: 
for $J \subset E(\cC)$, $r(J) = \dim(\cC|_J)$.

Two length-$n$ linear codes $\cC$ and $\cC'$ over $\F_q$ 
are defined to be \emph{equivalent} if there is an $n \times n$ 
permutation matrix $\Pi$ and an invertible $n \times n$ diagonal
matrix $\Delta$, such that $\cC'$ is the image of $\cC$ under
the vector space isomorphism $\phi: \F_q^n \rightarrow \F_q^n$ 
defined by $\phi(\x) = (\Pi \Delta) \x$. Informally, $\cC'$ is 
equivalent to $\cC$ if $\cC'$ can be obtained by first multiplying
the coordinates of $\cC$ by some nonzero elements of $\F_q$,
and then applying a coordinate permutation. 
In such a case, we write $\cC \equiv \cC'$. 
The equivalence class of codes equivalent
to $\cC$ will be denoted by $[\cC]$. It is clear that if
codes $\cC$ and $\cC'$ are equivalent, then their associated
matroids are isomorphic. 

We remark that code equivalence has been defined above according to the 
coding-theoretic convention. Note that, under this definition,
if $\cC'$ is obtained by applying an automorphism of the 
field $\F_q$ to $\cC$, then $\cC$ and $\cC'$ would in 
general be considered to be inequivalent. 
% Therefore, equivalent representations of 
% an $\F_q$-representable matroid $M$ 
% (in the sense of \cite[Section~6.3]{oxley})
% need not be equivalent as codes. However, since $\F_q$
% has at most $\log_2 q$ automorphisms, the two notions of equivalence
% differ by a factor of at most $\log_2 q$. To make this precise,
% if $M = M(\cC)$ for some code $\cC$ over $\F_q$, then 
% the number of equivalent $\F_q$-representations of $M$ 
% is at most $|[\cC]| \log_2 q $.

A family, $\mfC$, of codes over $\F_q$ is said to be
\emph{minor-closed} if, for each $\cC \in \mfC$, any code equivalent
to a minor of $\cC$ is also in $\mfC$. A code, $\cD$, over $\F_q$
is said to be an \emph{excluded minor} for a minor-closed family $\mfC$, 
if $\cD \notin \mfC$, but every proper minor of $\cD$ is in $\mfC$. 
It is easily verified that if $\mfC$ is a minor-closed
family, then a code $\cC$ is in $\mfC$ iff no minor of $\cC$ 
is an excluded minor for $\mfC$. 

Given a collection, $\cM$, of $\F_q$-representable matroids,
define the code family
\begin{equation}
\mfC(\cM) = \{\cC:\ \cC \text{ is a linear code over $\F_q$ such that }
                               M(\cC) \in \cM \}.
\label{CM_def}
\end{equation}
Evidently, if $\cM$ is a minor-closed class of $\F_q$-representable
matroids, then $\mfC(\cM)$ is also minor-closed. 
In this case, if $\cF$ is the set of all excluded minors for $\cM$, then 
$\mfC(\cF)$ is the set of all excluded minors for $\mfC(\cM)$.

\subsection{Pathwidth, Trellis-width and Branchwidth\label{pw_section}}

The definitions in this section rely on the notion of the
connectivity function of a matroid. Let $M$ be a matroid with
ground set $E(M)$ and rank function $r_M$. 
Its \emph{connectivity function}, $\l_M$, is defined by 
$\l_M(X) = r_M(X) + r_M(E(M)-X) - r_M(E(M))$ for $X \subset E(M)$. 
Note that $\l_M(X) = \l_M(E(M) - X)$, and $\l_M(E(M)) 
= \l_M(\emptyset) = 0$.
It should be pointed out that in the matroid theory literature, 
the prevalent definition of the connectivity function adds a `+1' 
to the expression we have given. We have chosen not to follow suit in order 
that we can give a minimum-fuss definition of trellis-width below.

The connectivity function is non-negative, \emph{i.e.},
$\l_M(X) \geq 0$ for all $X \subset E(M)$, and submodular
\emph{i.e.}, $\l_M(X \cup Y) + \l_M(X \cap Y) \leq \l_M(X) + \l_M(Y)$
for all $X,Y \subset E(M)$. It is monotone 
under the action of taking minors --- if $N$ is a minor of $M$,
then for all $X \subset E(N)$, $\l_N(X) \leq \l_M(X)$.
Finally, the connectivity function of a matroid is identical
to that of its dual, \emph{i.e.}, $\l_M(X) = \l_{M^*}(X)$ for all
$X \subset E(M)$.

Given an ordering $(e_1,e_2,\ldots,e_n)$ of the elements of $M$, 
define the \emph{width} of the ordering to be 
$w_M(e_1,e_2,\ldots,e_n) = \max_{i \in [n]} \l_M(e_1,e_2,\ldots,e_i)$.
(For simplicity of notation, we use $\l_M(e_1,e_2,\ldots,e_i)$ instead
of $\l_M(\{e_1,e_2,\ldots,e_i\})$.) 
The \emph{pathwidth} of $M$ is defined as 
$\pw(M) = \min w_M(e_1,e_2,\ldots,e_n)$,
the minimum being taken over all orderings $(e_1,e_2,\ldots,e_n)$ of $E(M)$.
An ordering $(e_1,e_2,\ldots,e_n)$ of $E(M)$ such that 
$w_M(e_1,e_2,\ldots,e_n) = \pw(M)$
is called an \emph{optimal} ordering.

Since $\l_M \equiv \l_{M^*}$, it is clear that $\pw(M) = \pw(M^*)$.
Another useful and easily verifiable property of pathwidth is 
that, for matroids $M_1$ and $M_2$, the pathwidth of their
direct sum, $\pw(M_1 \oplus M_2)$, equals $\max\{\pw(M_1),\pw(M_2)\}$.
The property of pathwidth most important for our purposes is
stated in the following lemma. \\[-6pt]

\begin{lemma}
If $N$ is a minor of $M$, then $\pw(N) \leq \pw(M)$.
\label{minor_pathwidth_lemma}
\end{lemma}
\begin{proof}
Let $(e_1,\ldots,e_n)$ be an optimal ordering of $E(M)$.
It is enough to show the
result in the case when $N = M \del e_i$ or $N = M \con e_i$ 
for some $i \in [n]$. In such a case, consider the ordering 
$(e_1,\ldots,e_{i-1},e_{i+1},\ldots,e_n)$ of $E(N)$. For 
$j \in \{1,\ldots,i-1\}$, we have 
$\l_N(e_1,\ldots,e_j) \leq \l_M(e_1,\ldots,e_j)$.
For $j \in \{i+1,\ldots,n\}$, we have
\begin{eqnarray*}
\l_N(e_1,\ldots,e_{i-1},e_{i+1},\ldots,e_j) 
&=& \l_N(e_{j+1},\ldots,e_n) \\
&\leq& \l_M(e_{j+1},\ldots,e_n) \ = \  \l_M(e_1,\ldots,e_{j}).
\end{eqnarray*}
It follows that $w_N(e_1,\ldots,e_{i-1},e_{i+1},\ldots,e_n) 
\leq w_M(e_1,\ldots,e_n) = \pw(M)$, and hence, $\pw(N) \leq \pw(M)$.
\end{proof} \mbox{}\\[-6pt]

The \emph{trellis-width} of a linear code $\cC$ over $\F_q$
is defined to be $\tw(\cC) = \pw(M(\cC))$. For a discussion
of the motivation and practical implications of this 
definition, we refer the reader to \cite[Section~5]{vardy}.

The pathwidth of a matroid is an upper bound on its branchwidth,
a more well known measure of matroid complexity. The branchwidth
of a matroid is defined via cubic trees. A \emph{cubic tree} is 
a tree in which the degree of any vertex is either one
or three. The vertices of degree one are called \emph{leaves}.
A \emph{branch-decomposition} of a matroid $M$ is a cubic tree,
$T$, with $|E(M)|$ leaves, labelled in a one-to-one fashion by the
elements of $M$. Each edge $e$ of such a branch-decomposition $T$ 
connects two subtrees of $T$, so $T \del e$ has two components. We say that
edge $e$ \emph{displays} a subset $X \subset E(M)$ if $X$ is the set of
labels of leaves of one of the components of $T \del e$. The \emph{width}
of an edge $e$ of $T$ is defined to be $\l_M(X)$, where $X$ is one
of the label sets displayed by $e$. The \emph{width} of $T$
is the maximum among the widths of its edges. 

\begin{figure}[t]
\centering{\epsfig{file=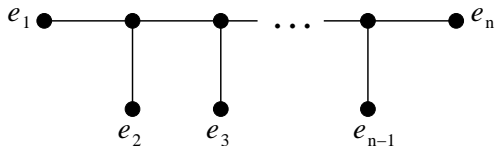, width=6.5cm}}
\caption{A branch-decomposition of $M$ having width equal to
$w_M(e_1,e_2,\ldots,e_n)$.}
\label{branch_decomp_fig}
\end{figure}

The \emph{branchwidth} of $M$ is the minimum among the widths
of all its branch-decompositions. Note that if $T$ is the 
branch-decomposition of $M$ shown in Figure~\ref{branch_decomp_fig},
then the width of $T$ is precisely $w_M(e_1,e_2,\ldots,e_n)$. Indeed,
the width of any edge of $T$ is either $\l_M(e_i)$ or 
$\l_M(e_1,\ldots,e_i)$ for some $i \in [n]$. Now, for any $i \in [n]$,
\begin{eqnarray*}
\l_M(e_1,\ldots,e_{i-1}) + \l_M(e_1,\ldots,e_i)
&=& \l_M(e_i,e_{i+1},\ldots,e_n) + \l_M(e_1,\ldots,e_i) \\
&\geq& \l_M(E(M)) + \l_M(e_i) \ = \ \l_M(e_i),
\end{eqnarray*}
the inequality above arising from the submodularity of $\l_M$.
Since $\l_M(e_i) \in \{0,1\}$, either 
$\l_M(e_1,\ldots,e_{i-1})$ or $\l_M(e_1,\ldots,e_i)$
is at least as large as $\l_M(e_i)$. Therefore, the width of 
$T$ is given by $\max_{i \in [n]} \l_M(e_1,\ldots,e_i)
= w_M(e_1,\ldots,e_n)$. It follows that the branchwidth of $M$ 
is upper-bounded by $\pw(M)$.

\section{NP-Hardness of Matroid Pathwidth and 
Code Trellis-Width\label{NP_section}}

In this section, we prove that for any fixed field $\F$, the
problem of computing the pathwidth of an $\F$-representable 
matroid $M$, given a representation of $M$ over $\F$, is NP-hard.
We accomplish this by reduction from the known NP-hard
problem of computing the pathwidth of a graph \cite{arnborg}, \cite{bod93}.

The notion of graph pathwidth was introduced by Robertson and Seymour
in \cite{RS-I}. Let $\cG$ be a graph with vertex set $V$.
An ordered collection $\cV = (V_1,\ldots,V_t)$, $t \geq 1$, of subsets of $V$ 
is called a \emph{path-decomposition} of $\cG$, if
\begin{itemize}
\item[(i)] $\bigcup_{i=1}^t V_i = V$;
\item[(ii)] for each pair of adjacent vertices $u,v \in V$, we have
$\{u,v\} \subset V_i$ for some $i \in [t]$; and
\item[(iii)] for $1 \leq i < j < k \leq t$, $V_i \cap V_k \subset V_j$.
\end{itemize}
The \emph{width} of such a path-decomposition $\cV$
is defined to be $w_{\cG}(\cV) = \max_{i \in [t]} |V_i| - 1$.
The \emph{pathwidth} of $\cG$, denoted by $\pw(\cG)$, is the 
minimum among the widths of all its path-decompositions. 
A path-decomposition $\cV$ such that $w_{\cG}(\cV) = \pw(\cG)$
is called an \emph{optimal} path-decomposition of $\cG$.

Let $\F$ be an arbitrary field. Given a graph $\cG$ with vertex set $V$, 
our aim is to produce, in time polynomial in $|V|$, a matrix $A$
over $\F$ such that $\pw(\cG)$ can be directly computed from $\pw(M[A])$. 
The NP-hardness of computing graph pathwidth then implies the 
NP-hardness of computing the pathwidth of an $\F$-representable matroid.

The obvious idea of taking $A$ to be a representation of
the cycle matroid of $\cG$ does not work. As observed by 
Robertson and Seymour \cite{RS-I}, trees can have arbitrarily 
large pathwidth; however, the cycle matroid of any tree is 
$U_{n,n}$ for some $n$, and $\pw(U_{n,n})= 0$.
What actually turns out to work is to take $A$ to be a 
representation of the cycle matroid of a certain graph 
constructible from $\cG$ in polynomial time, as we describe next.

\begin{figure}[t]
\centering{\epsfig{file=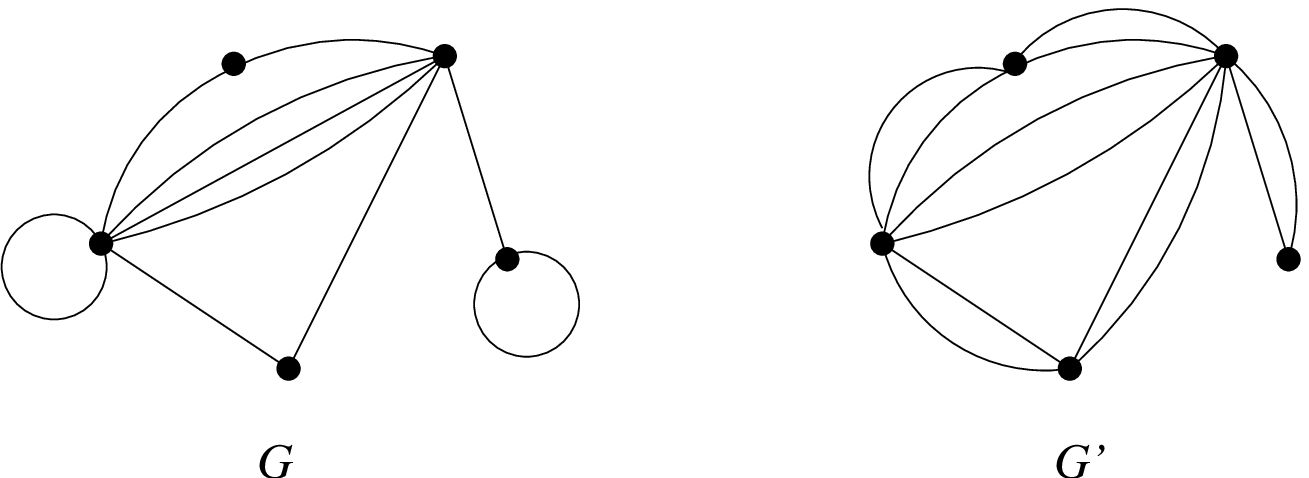, width=7.5cm}}
\caption{Construction of $\cG'$ from $\cG$.}
\label{Gprime}
\end{figure}

Let $\cG'$ be a graph defined on the same vertex set, $V$, as $\cG$, 
having the following properties (see Figure~\ref{Gprime}):
\begin{itemize}
\item[(P1)] $\cG'$ is loopless;
\item[(P2)] a pair of distinct vertices is adjacent in $\cG'$
iff it is adjacent in $\cG$; and
\item[(P3)] in $\cG'$, there are exactly two edges between 
each pair of adjacent vertices.
\end{itemize}
% Clearly, $\cG'$ can be constructed from $\cG$ in $O(|V|^2)$ time.
It is evident from the definition that 
$(V_1,\ldots,V_t)$ is a path-decomposition of $\cG$ iff it is
a path-decomposition of $\cG'$. Therefore, $\pw(\cG') = \pw(\cG)$.
% Thus, it is enough for our purposes to assume that $\cG$ itself
% is loopless and has property (P3) above.
% For the remainder of the section, we take $\cG$ to be a loopless 
% graph, with vertex set $V$, that has property (P3).

Define $\ocG$ to be the graph obtained by adding an extra vertex, 
henceforth denoted by $x$,
to $\cG'$, along with a pair of parallel edges from $x$ to 
each $v \in V$ (see Figure~\ref{Gbar}). Clearly, $\ocG$
is constructible directly from $\cG$ in $O(|V|^2)$ time.
% In particular, $V(\ocG)$ is the disjoint union
% of $V$ and $\{x\}$.
But more importantly, the pathwidth of the cycle matroid, $M(\ocG)$, 
of $\ocG$ relates very simply to the pathwidth of $\cG$, as
made precise by the following proposition. \\[-6pt]

\begin{proposition}
$\pw(M(\ocG)) = \pw(\cG) + 1$. \\[-6pt]
\label{pw_prop}
\end{proposition}

Before proving the result, we present some of its implications. 
For any field $\F$, $M(\ocG)$ is $\F$-representable. Indeed, 
if $D(\ocG)$ is any directed graph obtained by arbitrarily
assigning orientations to the edges of $\ocG$, then the 
vertex-arc incidence matrix of $D(\ocG)$ is an $\F$-representation
of $M(\ocG)$ \cite[Proposition~5.1.2]{oxley}. It is easily verified
that such an $\F$-representation of $M(\ocG)$ can be constructed
directly from $\cG$ in $O(|V|^3)$ time. Now, suppose that there
were a polynomial-time algorithm for computing the pathwidth
of an arbitrary $\F$-representable matroid, given an $\F$-representation
for it. Then, given any graph $\cG$, we can construct
an $\F$-representation, $A$, of $M(\ocG)$, and then compute the
pathwidth of $M[A] = M(\ocG)$, all in polynomial time. Therefore,
by Proposition~\ref{pw_prop}, we have a polynomial-time algorithm to
compute the pathwidth of $\cG$. However, the graph pathwidth problem 
is NP-hard. So, if there exists a polynomial-time algorithm for it,
then we must have $P=NP$. This implies the following result. \\[-6pt]

\begin{theorem}
Let $\F$ be a fixed field. The problem of computing the pathwidth of $M[A]$,
for an arbitrary matrix $A$ over $\F$, is NP-hard.
\\[-6pt]
\label{pw_NP_thm}
\end{theorem}

As a corollary, we have that computing the trellis-width of a
code is NP-hard. \\[-6pt]

\begin{corollary}
Let $\F$ be a fixed finite field. 
The problem of computing the trellis-width of an arbitrary linear code 
over $\F$, specified by any of its generator matrices, is NP-hard. \\[-6pt]
\label{tw_NP_cor}
\end{corollary}

\begin{figure}[t]
\centering{\epsfig{file=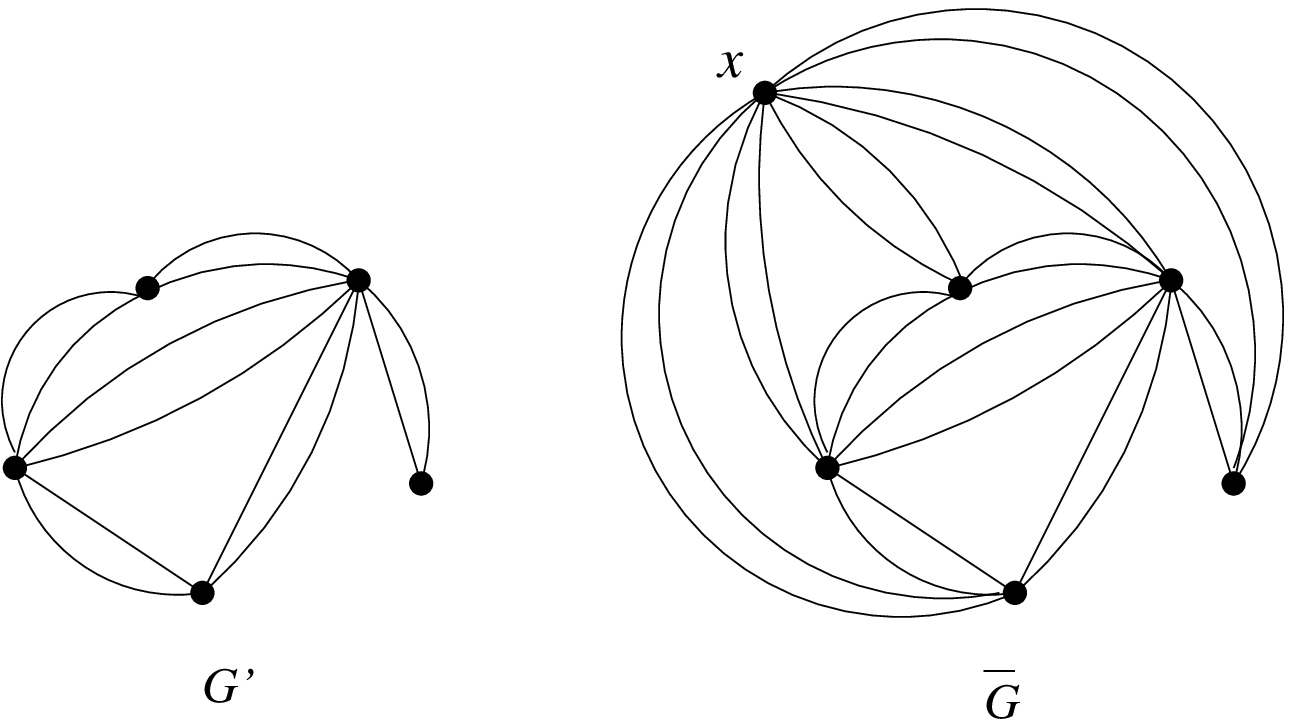, width=7.5cm}}
\caption{Construction of $\ocG$ from $\cG'$.}
\label{Gbar}
\end{figure}

The remainder of this section is devoted to 
the proof of Proposition~\ref{pw_prop}. Since 
$\pw(\cG') = \pw(\cG)$, for the purpose of our proof,
we may assume that $\cG' = \cG$. Thus, from now until the end of
this section, we take $\cG$ to be a loopless graph satisfying 
property (P3) above. Note that $\ocG$ also satisfies (P3). 
For each pair of adjacent vertices $u,v$ in $\cG$ or $\ocG$, 
we denote by $l_{uv}$ and $r_{uv}$ the 
two edges between $u$ and $v$. Let $V$ and $E$ 
denote the sets of vertices and edges of $\cG$, and
let $\oV$ and $\oE$ denote the corresponding sets of $\ocG$.
We thus have $\oV = V \disj \{x\}$, and 
$\oE = E \disj \left(\bigcup_{v \in V} \{l_{xv},r_{xv}\}\right)$.

Set $M = M(\ocG)$, so that $E(M) = \oE$. Note that since $\ocG$
is connected (each $v \in V$ is adjacent to $x$), we have
$\rank(M) = |\oV|-1 = |V|$. 

We will first prove that $\pw(M) \leq \pw(\cG) + 1$.
Let $\cV = (V_1,\ldots,V_t)$ be a path-decomposition of $\cG$.
We need the following fact about $\cV$: for each $j \in [t]$, 
\begin{equation}
\bigcup_{i \leq j} V_i \ \cap \ \bigcup_{k \geq j} V_k
\ = \  V_j.
\label{Vj_eq}
\end{equation}
The above equality follows from the fact 
that a path-decomposition, by definition, has the property that
for $1 \leq i < j < k \leq t$, $V_i \cap V_k \subset V_j$. 

For $j \in [t]$, let $F_j$ be the set of edges of $\cG$
that have both their end-points in $V_j$. By condition (ii) in
the definition of path-decomposition, $\bigcup_{j=1}^t F_j = E$. Now, let
$\overline{F_j} = F_j \cup \left(\bigcup_{v \in V_j} \{l_{xv},r_{xv}\}\right)$,
so that $\bigcup_{j=1}^t \overline{F_j} = \oE$. \\[-6pt]

\begin{Def}
An ordering $(e_1,\ldots,e_n)$ of the elements of a matroid $M$
is said to \emph{induce} an ordered partition $(E_1,\ldots,E_t)$
of $E(M)$ if for each $j \in [t]$, 
$\{e_{n_{j-1}+1},$ $e_{n_{j-1}+2},\ldots,e_{n_j}\} = E_j$,
where $n_j = \left|\bigcup_{i \leq j} E_j\right|$ (and $n_0 = 0$).
\\[-6pt]
\end{Def}

Let $\pi = (e_1,\ldots,e_n)$ be any ordering of $\oE$ that
induces the ordered partition $(E_1,E_2,\ldots,E_t)$, where
for each $j \in [t]$, 
$E_j = \overline{F_j} - \bigcup_{i < j} \overline{F_i}$.
We claim that the width of $\pi$ 
is at most one more than the width of the path-decomposition $\cV$. 
\\[-6pt]

\begin{lemma}
$w_M(\pi) \leq w_{\cG}(\cV) + 1$. \\[-6pt]
\label{width_lemma}
\end{lemma}
\begin{proof} 
Observe first that 
\begin{eqnarray}
w_M(\pi) &=& \max_{j \in [t]} \max_{1 \leq k \leq n_j-n_{j-1}} 
\l_M\left(\bigcup_{i < j} E_i \cup 
\{e_{n_{j-1}+1}, \ldots, e_{n_{j-1}+k}\}\right) \nonumber \\
& \leq & \max_{j \in [t]} \max_{E' \subset E_j} 
\l_M \left(\bigcup_{i < j} E_i \cup E'\right).
\label{wM_eq}
\end{eqnarray}

Let $X = \bigcup_{i < j} E_i \cup E'$ for some $j \in [t]$ and
$E' \subset E_j$, and consider $\l_M(X) = r_M(X) + r_M(\oE - X) - r_M(\oE)$.
Since $\ocG$ is a connected graph, $r_M(\oE) = |\oV|-1 = |V|$.

If $v$ is a vertex of $\ocG$ incident with an edge in $X$,
then $v \in \bigcup_{i \leq j} V_j \disj \{x\}$. So,
the subgraph of $\ocG$ induced by $X$ has its vertices contained
in $\bigcup_{i \leq j} V_j \disj \{x\}$. Therefore,
$r_M(X) \leq \left|\bigcup_{i \leq j} V_j \disj \{x\}\right| - 1
= \left|\bigcup_{i \leq j} V_j\right|$.

Next, consider $\oE - X = (\bigcup_{k > j} E_k) \cup (E_j - E')$.
Reasoning as above, the subgraph of $\cG$ induced by $\oE - X$
has its vertices contained in $\bigcup_{k \geq j} V_k \disj \{x\}$.
Hence, $r_M(\oE-X) \leq \left|\bigcup_{k \geq j} V_k\right|$.

Therefore, we have 
\begin{eqnarray*}
\l_M(X) & \leq & 
\left|\bigcup_{i \leq j} V_j\right| + \left|\bigcup_{k \geq j} V_k\right|
    - |V| \\
&=& \left|\bigcup_{i \leq j} V_j \cap \bigcup_{k \geq j} V_k\right| 
\ \ = \ \  |V_j|,
\end{eqnarray*}
the last equality arising from (\ref{Vj_eq}).
Hence, carrying on from (\ref{wM_eq}),
$$
w_M(\pi) \leq \max_{j \in [t]} |V_j| = w_{\cG}(\cV) + 1,
$$
as desired.
\end{proof} \mbox{} \\[-6pt]

The fact that $\pw(M) \leq \pw(\cG)+1$ easily follows from the above lemma.
Indeed, we may choose $\cV$ to be an optimal path-decomposition of $\cG$.
Then, by Lemma~\ref{width_lemma},
there exists an ordering $(e_1,\ldots,e_n)$ of $E(M)$ such that 
$w_M(e_1,\ldots,e_n) \leq \pw(\cG)+1$.
Hence, $\pw(M) \leq w_M(e_1,\ldots,e_n) \leq \pw(\cG)+1$. \\[-6pt]

We prove the reverse inequality in two steps, first showing that
$\pw(\ocG) = \pw(\cG)+1$, and then showing that $\pw(M) \geq \pw(\ocG)$.
\\[-6pt]

\begin{lemma}
$\pw(\ocG) = \pw(\cG) + 1$. \\[-6pt]
\label{pw_lemma}
\end{lemma}
\begin{proof}
Clearly, if $\cV = (V_1,\ldots,V_t)$ is a path-decomposition of $\cG$, then 
$\ocV = (V_1 \cup \{x\}, \ldots, V_t \cup \{x\})$ is a path-decomposition
of $\ocG$. Hence, choosing $\cV$ to be an optimal path-decomposition of $\cG$,
we have that $\pw(\ocG) \leq w_{\ocG}(\ocV) = w_{\cG}(\cV)+1 = \pw(\cG)+1$.

For the inequality in the other direction, we will show that 
there exists an optimal path-decomposition, $\tcV = (\tV_1, \ldots, \tV_s)$,
of $\ocG$ such that $x \in \tV_i$ for all $i \in [s]$. We then
have $\cV = (\tV_1 - \{x\}, \ldots, \tV_s - \{x\})$ being a
path-decomposition of $\cG$, and hence, 
$\pw(\cG) \leq w_{\cG}(\cV) = w_{\ocG}(\tcV)-1 = \pw(\ocG)-1$. 

Let $\ocV = (\oV_1,\ldots,\oV_t)$ be an optimal path-decomposition
of $\ocG$, and let $i_0 =\min\{i: x \in \oV_i\}$ and 
$i_1 = \max\{i: x \in \oV_i\}$. Since $\oV_i \cap \oV_k \subset \oV_j$
for $i < j < k$, we must have $x \in \oV_i$ for each $i \in [i_0,i_1]$.

We claim that $(\oV_{i_0}, \oV_{i_0+1}, \ldots, \oV_{i_1})$ is a
path-decomposition of $\ocG$. We only have to show that 
$\bigcup_{i=i_0}^{i_1} \oV_i = \oV$, and that for each pair
of adjacent vertices $u,v \in \oV$, $\{u,v\} \subset \oV_i$
for some $i \in [i_0,i_1]$. To see why the first assertion
is true, consider any $v \in \oV$, $v \neq x$. Since
$x$ is adjacent to $v$, and $\ocV$ is a path-decomposition of $\ocG$,
$\{x,v\} \subset \oV_i$ for some $i \in [t]$. However, $x \in \oV_i$
iff $i \in [i_0,i_1]$, and so, $\{x,v\} \subset \oV_i$ for some
$i \in [i_0,i_1]$. In particular, $v \in \oV_i$ for some $i \in [i_0,i_1]$.

For the second assertion, suppose that $u,v$ is a pair of vertices
adjacent in $\ocG$. Obviously, $\{u,v\} \subset \oV_j$ for some $j \in [t]$.
Suppose that $j \notin [i_0,i_1]$. We consider the case when $j > i_1$;
the case when $j < i_0$ is similar. As $\bigcup_{i=i_0}^{i_1} \oV_i = \oV$,
there exist $i_2, i_3 \in [i_0,i_1]$ such that $u \in \oV_{i_2}$ 
and $v \in \oV_{i_3}$. Without loss of generality (WLOG), $i_2 \leq i_3$. 
If $i_2 = i_3$, then there exists $i \in [i_0,i_1]$ such that 
$\{u,v\} \subset \oV_i$. If $i_2 < i_3$, we have 
$u \in \oV_{i_2} \cap \oV_j$ and $i_2 < i_3 < j$. Hence, $u \in \oV_{i_3}$
as well, and so once again, we have an $i \in [i_0,i_1]$ such that
$\{u,v\} \in \oV_i$.

Thus, $(\oV_{i_0}, \oV_{i_0+1}, \ldots, \oV_{i_1})$ is a path-decomposition
of $\ocG$, with the property that $x \in \oV_i$ for all $i \in [i_0,i_1]$.
It must be an optimal path-decomposition, since it is
a subsequence of the optimal path-decomposition $\ocV$.
\end{proof} \mbox{}\\[-6pt]

To complete the proof of Proposition~\ref{pw_prop}, it 
remains to show that $\pw(M) \geq \pw(\ocG)$. We introduce
some notation at this point. Recall that the two edges between
a pair of adjacent vertices $u$ and $v$ in $\ocG$ (or $\cG$) 
are denoted by $l_{uv}$ and $r_{uv}$. We define 
\begin{eqnarray*}
L_{\cG} &=& \{l_{uv}: u,v \text{ are adjacent vertices in } \cG\}, \\
R_{\cG} &=& \{r_{uv}: u,v \text{ are adjacent vertices in } \cG\},
\end{eqnarray*}
$L_x = \bigcup_{v \in V} \{l_{xv}\}$ and 
$R_x = \bigcup_{v \in V} \{r_{xv}\}$, where $x$ is the distinguished
vertex in $\oV - V$. Thus, $L_\cG \cup R_\cG = E$ and 
$E \cup L_x \cup R_x = \oE$.
Note that, by construction of $\ocG$, $\cl_M(L_x) = \cl_M(R_x) = \oE$,
where $\cl_M$ denotes the closure operator of $M$.

We will need the fact that there exists an optimal ordering 
$(e_1,\ldots,e_n)$ of $\oE$ that induces a certain ordered partition
of $\oE$ of the form 
$$
(L_1,A_1,B_1,R_1,L_2,A_2,B_2,R_2,\ldots,L_t,A_t,B_t,R_t),
$$
where for each $j \in [t]$, $L_j \subset L_x$, $A_j \subset L_\cG$,
$B_j \subset R_\cG$, and $R_j \subset R_x$. This will follow from a 
re-ordering argument given further below.
But first, we make some simple observations about orderings of $\oE$.
Given an ordering of $\oE$, we may assume, WLOG, that for each pair
of adjacent vertices $u,v \in \oV$, $l_{uv}$ appears before $r_{uv}$
in the ordering; we denote this by $l_{uv} < r_{uv}$. We call
such an ordering of $\oE$ a \emph{normal} ordering.
\\[-6pt]

\begin{lemma}
Let $(e_1,\ldots,e_n)$ be a normal ordering of $\oE$. Then, 
for $1 \leq j \leq n-1$, we have
\begin{itemize}
\item[(a)] $\l_M(e_1,\ldots,e_{j+1}) = \l_M(e_1,\ldots,e_j) + 1$
iff $e_{j+1} \notin \cl_M(e_1,\ldots,e_j)$; and
\item[(b)] $\l_M(e_1,\ldots,e_{j+1}) = \l_M(e_1,\ldots,e_j) - 1$
iff $e_{j+1} \notin \cl_M(e_{j+2},\ldots,e_n)$. \\[-6pt]
\end{itemize}
\label{conn_fn_lemma}
\end{lemma}
\begin{proof}
We only prove (a), as the proof of (b) is similar. It is easy to deduce 
from the definition of the connectivity function that
$\l_M(e_1,\ldots,e_{j+1}) = \l_M(e_1,\ldots,e_j) + 1$ iff 
$e_{j+1} \notin \cl_M(e_1,\ldots,e_j)$ and
$e_{j+1} \in \cl_M(e_{j+2},\ldots,e_n)$. 

Now, if $e_{j+1} \notin \cl_M(e_1,\ldots,e_j)$,
then $e_{j+1} = l_{uv}$ for some $u,v$. (If not, \emph{i.e.}, 
if $e_{j+1} = r_{uv}$, then since $l_{uv} < r_{uv}$, 
we must have $l_{uv} \in \{e_1,\ldots,e_j\}$, 
and so, $e_{j+1} = r_{uv} \in \cl_M(l_{uv}) \subset \cl_M(e_1,\ldots,e_j)$, 
a contradiction.) Therefore, $\{e_{j+2},\ldots,e_n\}$ contains
$r_{uv}$, and hence, $e_{j+1} = l_{uv} \in \cl_M(e_{j+2},\ldots,e_n)$.
We have thus shown that if $e_{j+1} \notin \cl_M(e_1,\ldots,e_j)$,
then $e_{j+1} \in \cl_M(e_{j+2},\ldots,e_n)$. Part (a) of
the lemma now follows. \end{proof} \mbox{} \\[-6pt]

We now describe a procedure that takes as input a normal ordering
of $\oE$, and produces as output 
a re-ordering of $\oE$ with certain desirable properties. \\[-6pt]

\pagebreak
\noindent\underline{\textsc{Re-ordering Algorithm}} \\[-6pt]

\emph{Input}: a normal ordering $(e_1,\ldots,e_n)$ of $\oE$. \\[-6pt]

\emph{Initialization}: $j = 0$. \\[-4pt]

\begin{tabular}{ll}
\underline{Step 0}: & If $j=0$, set $X_j = \emptyset$; \\
& else, set 
$X_j = \cl_M(e_1,\ldots,e_j) - \{e_1,\ldots,e_j\}$. \\[6pt]
\underline{Step 1}: & If $X_j = \emptyset$, \\
& \ \ \ find the least $k > j$ such that \\ 
& \ \ \ \ \ \ for some $m > j$, $e_m \in L_x \cap \cl_M(e_1,\ldots,e_k)$; \\
& \ \ \ set $(e_1',\ldots,e_n') = 
(e_1,\ldots,e_j,e_m,e_{j+1},\ldots,e_{m-1},e_{m+1},\ldots,e_n)$.\\[6pt]
& If $X_j \neq \emptyset$, \\
& \ \ \ if $L_x \cap X_j \neq \emptyset$, find an $m > j$
such that $e_m \in L_x \cap X_j$; \\
& \ \ \ else, if $L_\cG \cap X_j \neq \emptyset$, find an $m > j$
such that $e_m \in L_\cG \cap X_j$; \\
& \ \ \ else, if $R_\cG \cap X_j \neq \emptyset$, find an $m > j$
such that $e_m \in R_\cG \cap X_j$; \\
& \ \ \ else, if $R_x \cap X_j \neq \emptyset$, find an $m > j$
such that $e_m \in R_x \cap X_j$; \\
& \ \ \ set $(e_1',\ldots,e_n') = 
(e_1,\ldots,e_j,e_m,e_{j+1},\ldots,e_{m-1},e_{m+1},\ldots,e_n)$.\\[6pt]
\underline{Step 2}: & Replace $j$ by $j+1$. \\
& If $j < n$, replace $(e_1,\ldots,e_n)$ by $(e_1',\ldots,e_n')$, \\
& \ \ \ and return to Step 0; \\
& else, output $(e_1',\ldots,e_n')$.
\end{tabular}
\mbox{}\\[10pt]

Denote by $(e_1^*,\ldots,e_n^*)$ the final output generated by the 
above algorithm. Set $X_0^* = \emptyset$, and for 
$j \in [n]$, $X_j^* = \cl_M(e_1^*,\ldots,e_j^*) - \{e_1^*,\ldots,e_j^*\}$. 
Stepping through the algorithm, one may easily check that 
$(e_1^*,\ldots,e_n^*)$ has the following property:
for $0 \leq j \leq n-1$, if $X_j^* = \emptyset$, 
then $e_{j+1}^* \in L_x$, and if $X_j^* \neq \emptyset$, then 
$$
e_{j+1}^* \in 
\begin{cases}
L_x \cap X_j^*,  & \text{ if } L_x \cap X_j^* \neq \emptyset \\
L_\cG \cap X_j^*, & \text{ if } L_x \cap X_j^* = \emptyset, 
\text{ but } L_{\cG} \cap X_j^* \neq \emptyset \\
R_{\cG} \cap X_j^*, & \text{ if } L_x \cap X_j^* = L_{\cG} \cap X_j^* 
= \emptyset, \text{ but } R_{\cG} \cap X_j^* \neq \emptyset \\
R_x \cap X_j^*, & 
\text{ if } L_x \cap X_j^* = L_{\cG} \cap X_j^* = R_{\cG} \cap X_j^* 
= \emptyset, \text{ but } R_x \cap X_j^* \neq \emptyset. \\
\end{cases}
$$
The following claim can be readily deduced from this property,
and we leave the details to the reader. \\[-6pt]

\begin{claim} (a)\ The ordering $(e_1^*,\ldots,e_n^*)$ induces an 
ordered partition of $\oE$ of the form 
$$
(L_1,A_1,B_1,R_1,L_2,A_2,B_2,R_2,\ldots,L_t,A_t,B_t,R_t),
$$
where for each $j \in [t]$, $L_j \subset L_x$, $A_j \subset L_\cG$,
$B_j \subset R_\cG$ and $R_j \subset R_x$. Moreover,
for each $u,v \in \oV$, $l_{uv} \in L_j \cup A_j$ iff 
$r_{uv} \in B_j \cup R_j$. \\[6pt]
(b)\ For the ordered partition in (a), we have for each $j \in [t]$, 
$$
A_j \cup B_j \subset \cl_M(\bigcup_{i \leq j} L_i) - 
\cl_M(\bigcup_{i < j} L_i).
$$
% \{e_1^*,\ldots,e_j^*\} \subset \cl_M(\{e_i^*: e_i^* \in L_x, i \leq j\})$.
\\[-6pt]
\label{re-ordering_claim}
\end{claim}

The crucial property of $(e_1^*,\ldots,e_n^*)$ is the following.
\\[-6pt]

\begin{lemma}
If $(e_1^*,\ldots,e_n^*)$ is the output of the Re-ordering Algorithm
in response to the input $(e_1,\ldots,e_n)$, then
$w_M(e_1^*,\ldots,e_n^*) \leq w_M(e_1,\ldots,e_n)$. \\[-6pt]
\label{re-ordering_lemma}
\end{lemma}
\begin{proof}
Steps~0--1 of the algorithm go through $n$ iterations,
indexed by $j \in \{0,1,\ldots,n-1\}$. In the $j$th iteration,
Step 1 is given a normal ordering $(e_1,\ldots,e_n)$, in 
response to which it produces an ordering $(e_1',\ldots,e_n')$,
which is also normal. To prove the lemma, it is enough
to show that $w_M(e_1',\ldots,e_n') \leq w_M(e_1,\ldots,e_n)$.

So, suppose that the algorithm is in its $j$th iteration 
($0 \leq j \leq n-1$). We first dispose of the case when 
$X_j \neq \emptyset$. Then, 
$$
(e_1',\ldots,e_n') = 
(e_1,\ldots,e_j,e_m,e_{j+1},\ldots,e_{m-1},e_{m+1},\ldots,e_n)
$$
for some $m > j$ such that $e_m \in X_j$. Observe that if
$1 \leq s \leq j$ or if $m \leq s \leq n$, then 
$(e_1',\ldots,e_s')$ is just a re-ordering of $(e_1,\ldots,e_s)$,
and hence $\l_M(e_1',\ldots,e_s') = \l_M(e_1,\ldots,e_s)$.

So, consider $j < s < m$. In this case, $(e_1',\ldots,e_s') = 
(e_1,\ldots,e_j,e_m,e_{j+1},\ldots,e_{s-1})$. Since 
$e_m \in \cl_M(e_1,\ldots,e_j)$, we have $r_M(e_1',\ldots,e_s')
= r_M(e_1,\ldots,e_j,e_{j+1},\ldots,e_{s-1})$. On the other hand,
\begin{eqnarray}
r_M(e_{s+1}',\ldots,e_n') &=& r_M(e_s,\ldots,e_{m-1},e_{m+1},\ldots,e_n) 
\nonumber \\
&\leq& r_M(e_s,\ldots,e_{m-1},e_m,e_{m+1},\ldots,e_n).
\label{rM_eq1}
\end{eqnarray}
Hence, $\l(e_1',\ldots,e_s') \leq \l(e_1,\ldots,e_{s-1})$.
Therefore, for any $s \in [n]$, we have shown that there
exists a $t \in [n]$ such that $\l(e_1',\ldots,e_s') 
\leq \l(e_1,\ldots,e_t)$. It follows that 
$w_M(e_1',\ldots,e_n') \leq w_M(e_1,\ldots,e_n)$. \\[-6pt]
% \begin{eqnarray*}
% w_M(e_1',\ldots,e_n') &=& \max_{s \in [n]} \l(e_1',\ldots,e_s') \\
% &\leq& \max_{t \in [n]} \l(e_1,\ldots,e_t) \ = \ w_M(e_1,\ldots,e_n).
% \end{eqnarray*}

We must now deal with the case when $X_j = \emptyset$, 
\emph{i.e.}, $\cl_M(e_1,\ldots,e_j) = \{e_1,\ldots,e_j\}$. 
Note that if $L_x \subset \{e_1,\ldots,e_j\}$, then since 
$\cl_M(L_x) = \oE$, we have $\cl_M(e_1,\ldots,e_j) = \oE$. 
Therefore, $\{e_1,\ldots,e_j\} = \oE$, 
which means that $j = n$, a contradiction. Therefore, 
there must exist some $m > j$ such that $e_m \in L_x$. 

Let $k^*$ be the least integer $k > j$ such that there exists 
$e_m \in L_x \cap \cl_M(e_1,\ldots,e_k)$ for some $m > j$. 
By choice of $k^*$, we have $m \geq k^*$. For this $m$, we again have
$$
(e_1',\ldots,e_n') = 
(e_1,\ldots,e_j,e_m,e_{j+1},\ldots,e_{m-1},e_{m+1},\ldots,e_n).
$$
As before, if $1 \leq s \leq j$ or if $m \leq s \leq n$, then
$\l_M(e_1',\ldots,e_s') = \l_M(e_1,\ldots,e_s)$. For $k^* < s < m$,
we have
$$
r_M(e_1',\ldots,e_s') = r_M(e_1,\ldots,e_j,e_m,e_{j+1},\ldots,e_{s-1})
= r_M(e_1,\ldots,e_j,e_{j+1},\ldots,e_{s-1}),
$$
as $e_m \in \cl_M(e_1,\ldots,e_{k^*}) \subset \cl_M(e_1,\ldots,e_{s-1})$.
And as in (\ref{rM_eq1}), $r_M(e_{s+1}',\ldots,e_n') \leq 
r_M(e_s,\ldots,e_n)$. Hence, 
$\l(e_1',\ldots,e_s') \leq \l(e_1,\ldots,e_{s-1})$.

We are left with $j+1 \leq s \leq k^*$. Note that by choice of $k^*$,
$e_m \notin \cl_M(e_1,\ldots,e_{s-1})$. Therefore,
\begin{eqnarray*}
r_M(e_1',\ldots,e_s') &=& r_M(e_1,\ldots,e_j,e_m,e_{j+1},\ldots,e_{s-1}) \\
&=& 1 + r_M(e_1,\ldots,e_j,e_{j+1},\ldots,e_{s-1}),
\end{eqnarray*}
Since (\ref{rM_eq1}) again applies, we have that 
\begin{equation}
\l_M(e_1',\ldots,e_s') \leq 1 + \l_M(e_1,\ldots,e_{s-1}).
\label{l_eq1}
\end{equation}
Observe that, since $e_{j+1} \notin \cl_M(e_1,\ldots,e_j)$, 
by Lemma~\ref{conn_fn_lemma}(a), 
\begin{equation}
\l_M(e_1,\ldots,e_{j+1}) = \l_M(e_1,\ldots,e_j) + 1.
\label{l_eq2}
\end{equation}
Furthermore, by choice of $k^*$, 
$e_m \notin \cl_M(e_1,\ldots,e_{k^*-1})$, but 
$e_m \in \cl_M(e_1,\ldots,e_{k^*})$, which together imply that 
$e_{k^*} \notin \cl_M(e_1,\ldots,e_{k^*-1})$. Hence, again
by Lemma~\ref{conn_fn_lemma}(a),
\begin{equation}
\l_M(e_1,\ldots,e_{k^*}) = \l_M(e_1,\ldots,e_{k^*-1}) + 1.
\label{l_eq3}
\end{equation}
Therefore, from (\ref{l_eq1})--(\ref{l_eq3}), we find that
for $s = j+1$ or $s=k^*$, we have $\l_M(e_1',\ldots,e_s')
\leq \l_M(e_1,\ldots,e_s)$.

We claim that for $j+1 < s < k^*$, we have 
$\l_M(e_1,\ldots,e_{s-1}) \leq \l_M(e_1,\ldots,e_s)$, so that
by induction, $\l_M(e_1,\ldots,e_{s-1}) \leq \l_M(e_1,\ldots,e_{k^*-1})$.
This would then imply, via (\ref{l_eq1}) and (\ref{l_eq3}), that
$\l_M(e_1',\ldots,e_s') \leq \l_M(e_1,\ldots,e_{k^*})$.
Thus, for any $s \in [n]$, we have a $t \in [n]$ such that 
$\l(e_1',\ldots,e_s') \leq \l(e_1,\ldots,e_t)$. Therefore,
$w_M(e_1',\ldots,e_n') \leq w_M(e_1,\ldots,e_n)$, which would
complete the proof of the lemma. 

To prove our claim, it is enough to show that when $j+1 < s < k^*$,
we have $e_s \in \cl_M(e_{s+1},\ldots,e_n)$.
Indeed, it then follows from Lemma~\ref{conn_fn_lemma}(b)
that $\l_M(e_1,\ldots,e_s) \geq \l_M(e_1,\ldots,e_{s-1})$.
So, suppose that $e_s \notin \cl_M(e_{s+1},\ldots,e_n)$
for some $j+1 < s < k^*$. Then, $e_s = r_{uv}$ for some $u,v \in \oV$.
(Otherwise, if $e_s = l_{uv}$, then since $r_{uv} > l_{uv}$,
we would have $e_s \in \cl_M(e_{s+1},\ldots,e_n)$.) Note that
$l_{uv} \notin \{e_1,\ldots,e_j\}$; otherwise, the fact that
$\{e_1,\ldots,e_j\}$ is a flat of $M$ would imply that 
$r_{uv} \in \{e_1,\ldots,e_j\}$. 
So, $l_{uv} \in \{e_{j+1},\ldots,e_{s-1}\}$.

Suppose that $e_s = r_{xv}$ for some $v \in \oV$. Then, 
$l_{xv} \in \{e_{j+1},\ldots,e_{s-1}\}$, which contradicts the 
choice of $k^*$. Therefore, $e_s \notin R_x$, meaning that
$e_s = r_{uv}$ for some $u,v \in V$.

Now, if $l_{xu},l_{xv} \in \{e_{s+1},\ldots,e_n\}$, then
$e_s \in \cl_M(e_{s+1},\ldots,e_n)$, as $(l_{xu},l_{xv},r_{uv})$ is
a triangle in $\ocG$. So, WLOG, $l_{xu} \notin \{e_{s+1},\ldots,e_n\}$.
By choice of $k^*$, $l_{xu} \notin \{e_{j+1},\ldots,e_s\}$. Therefore,
$l_{xu} \in \{e_1,\ldots,e_j\}$. But now, 
$l_{xv} \in \cl_M(e_1,\ldots,e_{s-1})$, as $(l_{xu},l_{uv},l_{xv})$
is a triangle in $\ocG$. However, $l_{xv} \notin \{e_1,\ldots,e_j\}$;
otherwise, we would have $l_{xu}, l_{xv} \in \{e_1,\ldots,e_j\}$,
which, since $\{e_1,\ldots,e_j\}$ is a flat and $(l_{xu},l_{uv},l_{xv})$
is a triangle, would imply that $l_{uv} \in \{e_1,\ldots,e_j\}$.
Thus, $l_{xv} = e_{m^*}$ for some $m^* > j$. As already noted,
$l_{xv} \in \cl_M(e_1,\ldots,e_{s-1})$, and so once again,
our choice of $k^*$ is contradicted. 

Therefore, our assumption that $e_s \notin \cl_M(e_{s+1},\ldots,e_n)$
always leads to a contradiction, from which we conclude that 
the assumption is false. This completes the proof of the lemma.
\end{proof} \mbox{} \\[-6pt]

We can now furnish the last remaining piece of the proof of 
Proposition~\ref{pw_prop}. \\[-6pt]

\begin{lemma}
$\pw(M) \geq \pw(\ocG)$. \\[-6pt]
\label{pw_lemma2}
\end{lemma}
\begin{proof}
Let $(e_1,\ldots,e_n)$ be an optimal ordering of $\oE$. WLOG,
$(e_1,\ldots,e_n)$ may be assumed to be normal. Let $(e_1^*,\ldots,e_n^*)$
be the output of the Re-ordering Algorithm in the response
to the input $(e_1,\ldots,e_n)$. Then, $(e_1^*,\ldots,e_n^*)$ has
the properties listed in Claim~\ref{re-ordering_claim}, and, by 
Lemma~\ref{re-ordering_lemma}, is also an optimal ordering of $\oE$.

Now, $(e_1^*,\ldots,e_n^*)$ induces an ordered partition 
$(L_1,A_1,B_1,R_1,\ldots,L_t,A_t,B_t,R_t)$ of $\oE$, as in 
Claim~\ref{re-ordering_claim}(a). For $j \in [t]$, define $Y_j
= \bigcup_{i < j} (L_i \cup A_i \cup B_i \cup R_i) \cup (L_j \cup A_j)$,
and $Y_j' = \oE - Y_j \ = 
\bigcup_{i > j} (L_i \cup A_i \cup B_i \cup R_i) \cup (B_j \cup R_j)$.
Letting $\ocG[Y_j]$ and $\ocG[Y_j']$ denote the subgraphs of $\ocG$
induced by $Y_j$ and $Y_j'$, respectively, set 
$V_j = V(\ocG[Y_j]) \cap V(\ocG[Y_j'])$. In other words, 
$V_j$ is the set of vertices common to both $\ocG[Y_j]$ and $\ocG[Y_j']$.
It is easily checked that $\cV = (V_1,\ldots,V_t)$
is a path-decomposition of $\ocG$. Note that 
$$
|V_j| = |V(\ocG[Y_j])| + |V(\ocG[Y_j'])| - |\oV|.
$$

We next observe that $\ocG[Y_j]$ and $\ocG[Y_j']$ are connected graphs.
From Claim~\ref{re-ordering_claim}(b), we have that 
$Y_j \subset \cl_M(\bigcup_{i \leq j} L_i)$. Therefore, 
for any edge $l_{uv}$ (or $r_{uv}$) in $Y_j - \bigcup_{i \leq j} L_i$, 
both $l_{xu}$ and $l_{xv}$ must be in some $L_i$, $i \leq j$.
Thus, in $\ocG[Y_j]$, each vertex $v \neq x$ is adjacent to $x$,
which shows that $\ocG[Y_j]$ is connected. 

Consider any vertex $v \neq x$ in $\ocG[Y_j']$, such that 
$r_{xv} \notin Y_j'$. Then, $r_{uv} \in Y_j'$ for some $u \neq x$.
So, $r_{uv} \in B_k$ for some $k \geq j$. By Claim~\ref{re-ordering_claim}(b),
$r_{uv} \in \cl_M(\bigcup_{i \leq k} L_i) - \cl_M(\bigcup_{i < k} L_i)$.
This implies that either $l_{xu} \in L_k$ or $l_{xv} \in L_k$. Hence,
either $r_{xu} \in R_k$ or $r_{xv} \in R_k$. However, $r_{xv}$ cannot
be in $R_k$, since $r_{xv} \notin Y_j'$, and so, $r_{xu} \in R_k$. Thus,
$(r_{xu},r_{uv})$ forms a path in $\ocG[Y_j']$ from $x$ to $v$. 
It follows that $\ocG[Y_j']$ is connected.

Therefore,
\begin{eqnarray*}
\l_M(Y_j) &=& r_M(Y_j) + r_M(Y_j') - r_M(\oE) \\
&=& (|V(\ocG[Y_j])| - 1) + (|V(\ocG[Y_j'])| - 1) - (\oV-1)
= |V_j| - 1.
\end{eqnarray*}
Hence, 
%\begin{eqnarray*}
$$
\pw(M) = w_M(e_1^*,\ldots,e_n^*) 
\geq \max_{j \in [t]} \l_M(Y_j) 
= \max_{j \in [t]} |V_j| - 1 = w_\ocG(\cV) \geq \pw(\ocG),
$$
% \end{eqnarray*}
which proves the lemma.
\end{proof} \mbox{} \\[-6pt]

The proof of Proposition~\ref{pw_prop} is now complete.

\section{Matroids of Bounded Pathwidth\label{bounded_pw_section}}

Theorem~\ref{pw_NP_thm} shows that the following decision 
problem is NP-complete. \\[-6pt]

\begin{tabular}{rl}
\textbf{Problem:} & \textsc{Matroid Pathwidth} \\
\multicolumn{2}{l}{Let $\F$ be a fixed field.} \\[2pt]
\textbf{Instance:} & An $m \times n$ matrix $A$ over $\F$, and 
an integer $w > 0$. \\
\textbf{Question:} & Is there an ordering $(e_1,\ldots,e_n)$
of the elements of $M = M[A]$, \\
& such that $w_M(e_1,\ldots,e_n) \leq w$? \\[4pt]
\end{tabular}

\noindent Similarly, Corollary~\ref{tw_NP_cor} shows that the
corresponding decision problem for code trellis-width (over
a fixed finite field $\F$) is NP-complete. 

In this section, we consider the situation when the parameter $w$ above
is a fixed constant, and therefore, not considered to be part of 
the problem instance. In contrast to the NP-completeness of 
\textsc{Matroid Pathwidth},
we believe that the following decision problem and its
coding-theoretic counterpart are solvable in polynomial time. \\[-6pt]

\begin{tabular}{rl}
\textbf{Problem:} & \textsc{Weak Matroid Pathwidth} \\
\multicolumn{2}{l}{Let $\F_q = GF(q)$ be a fixed finite field, 
and $w$ a fixed positive integer.} \\[2pt]
\textbf{Instance:} & An $m \times n$ matrix $A$ over $\F_q$. \\
\textbf{Question:} & Is there an ordering $(e_1,\ldots,e_n)$
of the elements of $M = M[A]$, \\
& such that $w_M(e_1,\ldots,e_n) \leq w$? \\[4pt]
\end{tabular}

Our optimism above stems from the fact that 
the property of having pathwidth bounded by $w$ is 
preserved by the minors of a matroid. 
To be precise, let $\cP_{w,q}$ be the class of matroids 
representable over the finite field $\F_q = GF(q)$, that have
pathwidth at most $w$. By Lemma~\ref{minor_pathwidth_lemma},
$\cP_{w,q}$ is minor-closed. Since pathwidth is an upper
bound on the branchwidth of a matroid, all matroids in $\cP_{w,q}$
have branchwidth at most $w$. Now, Geelen and Gerards have shown 
that if $\cM$ is any minor-closed class of $\F_q$-representable matroids 
having bounded branchwidth, then $\cM$ has finitely many excluded minors
\cite[Theorem~1.4]{GW02}. As a result, we have the following theorem.
\\[-6pt]

\begin{theorem}
For any integer $w > 0$ and finite field $\F_q$,
$\cP_{w,q}$ has finitely many excluded minors.
Consequently, the code family
$$
\mfC(\cP_{w,q}) = \{\cC:\ \cC \text{ is a linear code over $\F_q$ such that }
                               \tw(\cC) \leq w\}.
$$
also has finitely many excluded minors. \\[-6pt]
\label{Pwq_thm}
\end{theorem}

Theorem~\ref{Pwq_thm} shows that deciding whether or not a given 
$\F_q$-representable matroid $M$ belongs to $\cP_{w,q}$ can be 
accomplished by testing whether or not $M$ contains as a minor 
one of the finitely many excluded minors of $\cP_{w,q}$. The 
Minor-Recognition Conjecture of Geelen, Gerards and Whittle 
\cite[Conjecture~1.3]{GGW} states that, for any fixed $\F_q$-representable
matroid $N$, testing a given $\F_q$-representable matroid for the 
presence of an $N$-minor can be done in polynomial time. 
So, if this conjecture is true --- and there is evidence to
support its validity \cite{GGW} --- then membership
of an $\F_q$-representable matroid in the class $\cP_{w,q}$ can be
decided in polynomial time. Hence, assuming the validity of the
Minor-Recognition Conjecture, \textsc{Weak Matroid Pathwidth}
is solvable in polynomial time.

While the finiteness of the list of excluded minors for $\cP_{w,q}$
implies, modulo the Minor-Recognition Conjecture, the existence
of a polynomial-time algorithm for \textsc{Weak Matroid Pathwidth},
an actual implementation of such an algorithm would 
require the explicit determination of the excluded minors. As a
relatively easy exercise, we prove the following theorem.\\[-6pt]

\begin{theorem}
A matroid is in $\cP_{1,q}$ iff it contains no minor isomorphic
to any of the matroids $U_{2,4}$, $M(K_4)$, $M(K_{2,3})$ and $M^*(K_{2,3})$.
\\[-6pt]
\label{P1q_thm}
\end{theorem}

We first verify the easy ``only if'' part of the above theorem.
\\[-6pt]

\begin{proposition}
$U_{2,4}$, $M(K_4)$, $M(K_{2,3})$ and $M^*(K_{2,3})$
are not in $\cP_{1,q}$. \\[-6pt]
\label{TC1_prop1}
\end{proposition}
\begin{proof}
If $(e_1,e_2,e_3,e_4)$ is any ordering of the elements of $M = U_{2,4}$,
then $\l_M(e_1,e_2) = r_M(e_1,e_2) + r_M(e_3,e_4) - \rank(M)
= 2+2-2 = 2$. It follows that $\pw(U_{2,4}) = 2$. 

Now consider $M = M(K_4)$. For any ordering $(e_1,\ldots,e_6)$ of 
$E(K_4)$, we have $r_M(e_1,e_2,e_3) \geq 2$, with equality 
iff $\{e_1,e_2,e_3\}$ is a triangle, in which case $\{e_4,e_5,e_6\}$ 
is a triad. It follows that $r_M(e_1,e_2,e_3) + r_M(e_4,e_5,e_6) \geq 5$. 
Hence, $w_M(e_1,\ldots,e_6) \geq \l_M(e_1,e_2,e_3) \geq 2$.

The proof for $M = M(K_{2,3})$ is very similar.
For any $J \subset E(K_{2,3})$ with $|J|=3$, $r_M(J) = 3$, since
$K_{2,3}$ has no circuits of size less than 4.
Therefore, for any ordering $(e_1,\ldots,e_6)$ of $E(K_{2,3})$,
$w_M(e_1,\ldots,e_6) \geq \l_M(e_1,e_2,e_3) = 3+3 - 4 = 2$. 
Thus, $\pw(M(K_{2,3})) \geq 2$, and by duality, 
$\pw(M^*(K_{2,3})) \geq 2$ as well.
\end{proof} \mbox{} \\[-6pt]

We now prove the ``if'' part of Theorem~\ref{P1q_thm}. For the duration
of the proof, we take $M$ to be a matroid that contains 
no minor isomorphic to the matroids listed in the statement of the 
theorem. Since $M(K_4)$ is a minor of each of the matroids
$F_7$, $F_7^*$, $M(K_5)$, $M^*(K_5)$, $M(K_{3,3})$ and $M^*(K_{3,3})$,
$M$ contains none of these as minors. Therefore,
$M = M(\cG)$ for some planar graph $\cG$ 
(cf.\ \cite[Theorem~13.3.1 and Proposition~5.2.6]{oxley}). Evidently,
we may take $\cG$ to be connected as a graph.

Since $\cP_{1,q}$ is closed under direct sums, we may assume
that $M$ is 2-connected. Therefore, $\cG$ is either a graph 
consisting of a single vertex with a self-loop incident with it, 
or $\cG$ is a loopless graph. In the former case, 
$M \cong U_{0,1}$, which is in $\cP_{1,q}$. So, we may 
assume that $\cG$ is loopless. If $\cG$ has exactly two vertices,
then $M \cong U_{1,n}$ for some $n$, which is also in $\cP_{1,q}$.
Hence, we may assume that $|V(\cG)| \geq 3$, in which case,
$\cG$ is 2-connected as a graph \cite[Corollary~8.2.2]{oxley}.
Moreover, if $\cG^*$ is any geometric dual of $\cG$,
then, since $M^* = M(\cG^*)$ is 2-connected, by the same
argument as above, we may assume that $\cG^*$ is also 
2-connected as a graph.

\begin{figure}[t]
\centering{\epsfig{file=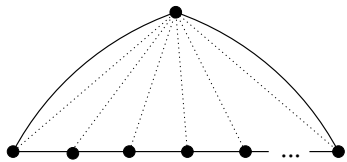, width=4.5cm}}
\caption{An ``umbrella'' graph. A dotted line between a pair
of vertices represents zero or more parallel edges between them.}
\label{umbrella}
\end{figure}

At this point, we need the following definition. We call a graph
an \emph{umbrella} if it is of the form shown in Figure~\ref{umbrella}.
Formally, an umbrella is a graph $H$ that consists of a circuit
on $m+1$ vertices $u_0,u_1,\ldots,u_m$, and in addition, for each 
$i \in [m]$, zero or more parallel edges between $u_0$ and $u_i$.
Note that $H - u_0$ is a simple path, where $H - u_0$ denotes
the graph obtained from $H$ by deleting the vertex $u_0$ and all
edges incident with it.

Returning to our proof, we have $M = M(\cG)$ for a loopless,
2-connected, planar graph $\cG$, such that any geometric dual
of $\cG$ is also 2-connected. \\[-6pt]

\begin{lemma}
$\cG$ has a geometric dual $\cG^*$ that is isomorphic to an umbrella. 
\\[-6pt]
\label{umb_lemma1}
\end{lemma}

We prove the lemma using the concept of an outerplanar graph. 
A planar graph is said to be \emph{outerplanar} if it has a 
planar embedding in which every vertex lies on the exterior 
(unbounded) face. We will refer to such a planar embedding of the graph
as an \emph{outerplanar embedding}. Outerplanar graphs were 
characterized by Chartrand and Harary \cite{CH67} as graphs that do not 
contain $K_4$ or $K_{2,3}$ as a minor. \\[-6pt]

\emph{Proof of Lemma~\ref{umb_lemma1}\/}: Since $M(\cG)$ contains no 
$M(K_4)$- or $M(K_{2,3})$-minor, 
$\cG$ cannot contain $K_4$ or $K_{2,3}$ as a minor. Therefore,
by the Chartrand-Harary result mentioned above, $\cG$ is outerplanar. 
Let $\cG^*$ be the geometric dual of an outerplanar embedding of $\cG$. 

Let $x$ be the vertex of $\cG^*$ corresponding to the exterior face
of the outerplanar embedding of $\cG$. By a result of Fleischner 
\emph{et al.}\ \cite[Theorem~1]{FGH74}, $\cG^* - x$ is a forest.
In fact, since $\cG^*$ is 2-connected, $\cG^* - x$ is a tree.

We claim that no vertex of $\cG^* - x$ has degree greater than two,
and hence, $\cG^* - x$ is a simple path. Indeed, suppose
that $\cG^* - x$ has a vertex $u$ adjacent to three
other vertices $v_1,v_2,v_3$. Since $G^*$ is 2-connected,
there are paths $\pi_1$, $\pi_2$ and $\pi_3$ in $\cG^*$ from
$v_1$, $v_2$ and $v_3$, respectively, to $x$ that do not pass 
through $u$. Also, since $\cG^* - x$ is a tree, these paths
must be internally disjoint in $\cG^*$. The graph $\cG^*$ thus has
a subgraph as depicted in Figure~\ref{deg3_obs}. But this subgraph 
is obviously contractible to $K_{2,3}$, and hence $\cG^*$ has $K_{2,3}$
as a minor. However, this is impossible, as $M^* = M(\cG^*)$
does not have $M(K_{2,3})$ as a minor. 

\begin{figure}[!t]
\centering{\epsfig{file=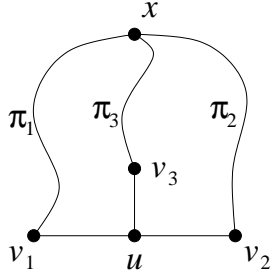, width=3.5cm}}
\caption{If $\cG^* - x$ has a vertex of degree at least 3,
then $\cG^*$ has a $K_{2,3}$ minor.}
\label{deg3_obs}
\end{figure}

Thus, $\cG^* - x$ is a simple path. The two degree-one vertices 
(end-points) of this path must be adjacent to $x$ in $\cG^*$;
otherwise, $\cG^*$ is not 2-connected. It follows that $\cG^*$ is 
isomorphic to an umbrella. \endproof \mbox{}\\[-6pt]

To complete the proof of Theorem~\ref{P1q_thm}, we 
show that $M(\cG^*) \in \cP_{1,q}$, so that by duality,
$M = M^*(\cG^*) \in \cP_{1,q}$. This is done by the following lemma. 

\begin{lemma}
If $H$ is an umbrella, then $M(H) \in \cP_{1,q}$.
\label{umb_lemma2}
\end{lemma}
\begin{proof}
Let $H$ be an umbrella on $m+1$ vertices $u_0, u_1, \ldots, u_m$,
where $u_0$ is the vertex such that $H - \{u_0\}$ is a simple path.
For $i \in [m]$, let $E_i$ denote the 
set of edges between $u_0$ and $u_i$. Also, for 
$j \in [m-1]$, let $e_j$ denote the edge between $u_j$ and $u_{j+1}$. 
Consider any ordering of the edges of $H$ that induces
the ordered partition
$$
(E_1,e_1,E_2,e_2,\ldots,E_{m-1},e_{m-1},E_m).
$$

Let $J = \left(\bigcup_{i=1}^{j-1} (E_i \cup \{e_i\})\right) \cup X$,
with $X \subset E_j$ ($X$ may be empty). 
Note that the subgraph, $H[J]$, of $H$ 
induced by the edges in $J$ is incident only with vertices in 
$\{u_0,u_1,\ldots,u_j\}$. Therefore, setting $M = M(H)$,
$r_M(J) = |V(H[J])|-1 \leq j$.
Similarly, the subgraph of $H$ induced by the edges in 
$E(H) - J$ is incident only with vertices in 
$\{u_j,u_{j+1},\ldots,u_m,u_0\}$, and so, $r_M(E(H)-J) \leq m-j+1$.

Thus, $\l_M(J) \leq j + (m-j+1) - m = 1$, and it follows 
that $\pw(M) \leq 1$. Being graphic, $M$ is $\F_q$-representable,
and hence, $M \in \cP_{1,q}$.
\end{proof} \mbox{} \\[-6pt]

This completes the proof of Theorem~\ref{P1q_thm}. \\[-6pt]

As a corollary to the theorem, we give a coding-theoretic
characterization of the code family $\mfC(\cP_{1,q})$. In
coding theory, an $\F_q$-representation of a uniform matroid is called
a \emph{maximum-distance separable (MDS) code}. For any field $\F_q$,
the matrices $G_4$, $G_{2,3}$ and $G_{2,3}^*$ below are $\F_q$-representations
of $M(K_4)$, $M(K_{2,3})$ and $M^*(K_{2,3})$, respectively.
$$
G_4 = 
\left[
\begin{array}{cccccc}
1  &  0  &  0  &  1  &  0  & -1 \\
0  &  1  &  0  &  1  &  1  & -1 \\
0  &  0  &  1  &  0  &  1  & -1 \\
\end{array}
\right];
$$
$$
G_{2,3} = 
\left[
\begin{array}{cccccc}
1  &  0  &  0  & 0 & -1  &  -1  \\
0  &  1  &  0  & 0 & 1  &  0  \\
0 & 0 & 1 & 0 & 0 & 1 \\
0 & 0 & 0 & 1 & 1 & 1 
\end{array}
\right]; \  \ 
G^*_{2,3} = 
\left[
\begin{array}{cccccc}
1  &  -1  &  0  & -1 & 1  &  0  \\
1  &  0  &  -1  & -1 & 0  &  1  \\
\end{array}
\right].
$$
The matroids $M(K_4)$, $M(K_{2,3})$ and $M^*(K_{2,3})$, being binary,
are uniquely representable over $\F_q$, in the matroid-theoretic sense 
\cite[Section~6.3 and Theorem~10.1.3]{oxley}. We let $\cC(K_4)$,
$\cC(K_{2,3})$ and $\cC(K_{2,3})^\perp$ denote the codes over $\F_q$
generated by the matrices $G_4$, $G_{2,3}$ and $G_{2,3}^*$, respectively.
\\[-6pt]

\begin{corollary}
Let $\F_q$ be an arbitrary finite field. A linear code $\cC$
over $\F_q$ has trellis-width at most one iff it contains
no minor equivalent to any of the following:
\begin{itemize}
\item[(i)] a $[4,2]$ MDS code;
\item[(ii)] a code obtainable by applying an automorphism of
$\F_q$ to one of the codes $\cC(K_4)$,
$\cC(K_{2,3})$ and $\cC(K_{2,3})^\perp$. \\[-6pt]
\end{itemize}
\label{CP1q_cor}
\end{corollary}

The problem of finding the complete set of excluded minors for 
$\cP_{w,q}$ quickly becomes difficult for $w > 1$. The
main obstacle is that we may only assume the basic property
of 2-connectedness for such excluded minors. The class 
$\cP_{w,q}$ is not even closed under 2-sums, so excluded minors
for the class need not be 3-connected. An illustration of this
is given by the following result, which provides a
partial list of excluded minors for $\cP_{2,q}$. \\[-6pt]

\begin{figure}[!t]
\centering{\epsfig{file=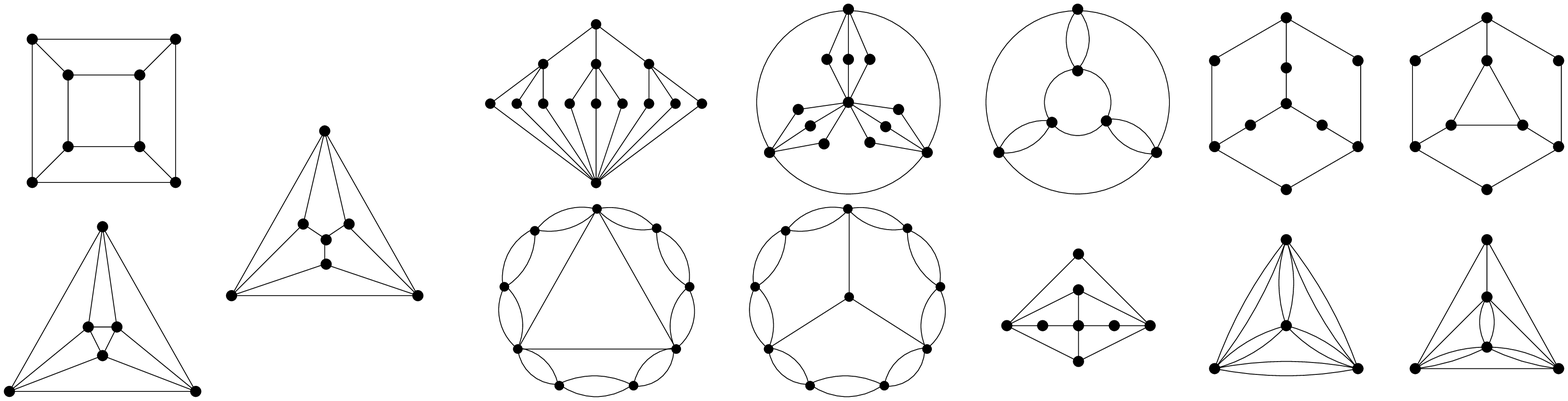, width=12cm}}
\caption{Some of the planar graphs whose cycle matroids 
are excluded minors for $P_{2,q}$.}
\label{P2q_ex_minors}
\end{figure}

\begin{proposition}
For any finite field $\F_q$, the matroids $F_7$, $F_7^*$,
$M(K_5)$, $M^*(K_5)$, $M(K_{3,3})$, $M^*(K_{3,3})$, and $M(\cG)$, 
where $\cG$ is any of the planar graphs in Figure~\ref{P2q_ex_minors}, 
are excluded minors for $\cP_{2,q}$. If $q \geq 4$, then 
$U_{3,6}$ is also an excluded minor for $\cP_{2,q}$. \\[-6pt]
\label{P2q_prop}
\end{proposition}

We omit the proof, as it is only a matter of verifying that for each 
matroid $M$ listed in the proposition, $M \notin \cP_{2,q}$, but
$M \del e, M \con e \in \cP_{2,q}$ for any $e \in E(M)$.
We point out that the cycle matroids of all but the three leftmost 
graphs in Figure~\ref{P2q_ex_minors} are not 3-connected. 

\section*{Acknowledgment} The author would like to thank Jim Geelen for
contributing some of his ideas to this paper, and Alexander Vardy for 
pointers to the prior literature on trellis complexity.

% Cut-and-paste buffer


\begin{thebibliography}{99}
\bibitem{arnborg} S.\ Arnborg, D.G.\ Corneil and A.\ Proskurowski,
``Complexity of finding embeddings in a $k$-tree,''
\emph{SIAM J.\ Alg.\ Disc.\ Meth.}, vol.\ 8, pp.\ 277--284, 1987.
% \bibitem{BW88} R.E.\ Bixby and D.K.\ Wagner, 
% ``An almost linear-time algorithm for graph realization,''
% \emph{Math.\ Oper.\ Res.}, vol.\ 13, no.\ 1, pp.\ 99--123, 1988.
\bibitem{bod93} H.L.\ Bodlaender, ``A tourist guide through treewidth,''
\emph{Acta Cybernetica}, vol.\ 11, pp.\ 1--23, 1993.
\bibitem{CH67} G.\ Chartrand and F.\ Harary, 
``Planar permutation graphs,'' 
\emph{Ann.\ Inst.\ Henri Poincar{\'e} Sec.\ B},
vol.\ III, no.\ 4, pp.\ 433--438, 1967.
\bibitem{FGH74} H.J.\ Fleischner, D.P.\ Geller and F.\ Harary,
``Outerplanar graphs and weak duals,'' \emph{J.\ Indian Math.\ Soc.},
vo.\ 38, pp.\ 215--219, 1974.
% \bibitem{For88} G.D.\ Forney Jr., ``Coset codes II: 
% Binary lattices and related codes,'' \emph{IEEE Trans.\ Inform.\ Theory}, 
% vol.\ 34, no.\ 5, pp.\ 1152--1187, Sept.\ 1988.
\bibitem{For94} G.D.\ Forney Jr., ``Dimension/length profiles and 
trellis complexity of linear block codes,'' 
\emph{IEEE Trans.\ Inform.\ Theory}, vol.\ 40, no.\ 6, pp.\ 1741--1752, 
Nov.\ 1994.
% \bibitem{For93} G.D.\ Forney Jr. and M.D.\ Trott, ``The dynamics of 
% group codes: State spaces, trellis diagrams and canonical encoders,'' 
% \emph{IEEE Trans.\ Inform.\ Theory}, vol.\ 39, no.\ 5, pp.\ 1491--1513, 
% Sept.\ 1993.
\bibitem{GGW} J.\ Geelen, B.\ Gerards and G.\ Whittle, 
``Towards a matroid-minor structure theory'', to appear in 
\emph{Combinatorics, Complexity and Chance. A tribute to Dominic Welsh}, 
G.\ Grimmett and C.\ McDiarmid, eds., Oxford University Press, 2007. 
Available online at \texttt{http://homepages.cwi.nl/$\sim$bgerards/personal/papers/towards\_welsh.pdf}.
%\texttt{http://www.math.uwaterloo.ca/$\sim$jfgeelen/survey.pdf}.
\bibitem{GGW06} J.\ Geelen, B.\ Gerards and G.\ Whittle, 
``On Rota's Conjecture and excluded minors containing large 
projective geometries,'' \emph{J.\ Combin.\ Theory, Ser.\ B}, vol.\ 96,
pp.\ 405--425, 2006.
\bibitem{GW02} J.\ Geelen and G.\ Whittle, 
``Branch-width and Rota's Conjecture,''
\emph{J.\ Combin.\ Theory, Ser.\ B}, vol.\ 86, no.\ 2,
pp.\ 315--330, Nov.\ 2002.
\bibitem{hall07} R.\ Hall, J.\ Oxley and C.\ Semple, 
``The structure of 3-connected matroids of path width three'', 
\emph{Europ.\ J.\ Combin.}, vol.\ 28, pp.\ 964--989, 2007.
% \bibitem{harary} F.\ Harary, \emph{Graph Theory}, Addison-Wesley,
% Reading, Mass., USA, 1969.
\bibitem{horn} G.B.\ Horn and F.R.\ Kschischang, 
``On the intractability of permuting a block code to 
minimize trellis complexity,''
\emph{IEEE Trans.\ Inform. Theory}, vol.\ 42, no.\ 6, pp.\ 2042--2048,
Nov.\ 1996.
\bibitem{jain} K.\ Jain, I.\ M{\u a}ndoiu and V.V.\ Vazirani,
``The `art of trellis decoding' is computationally hard --- 
for large fields,''
\emph{IEEE.\ Trans.\ Inform.\ Theory}, vol.\ 44, no.\ 3, 
pp.\ 1211--1214, May 1998.
\bibitem{kashyap} N.\ Kashyap, ``A decomposition theory for binary linear 
codes,'' submitted to \emph{IEEE Trans.\ Inform.\ Theory}. 
ArXiv e-print cs.DM/0611028.
% \bibitem{lafourcade} 
% A.\ Lafourcade and A.\ Vardy, ``Asymptotically good codes have infinite 
% trellis complexity,'' \emph{IEEE.\ Trans.\ Inform.\ Theory}, vol.\ 41,
% no.\ 2, pp.\ 555--559, March 1995.
\bibitem{sloane} F.J.\ MacWilliams and N.J.A.\ Sloane, {\em The Theory of
Error-Correcting Codes}, North-Holland, Amsterdam, 1977.
% \bibitem{mitchell} S.L.\ Mitchell, ``Linear algorithms to recognize
% outerplanar and maximal outerplanar graphs,'' \emph{Inf.\ Proc.\ Lett.},
% vol.\ 9, no.\ 5, pp.\ 229--232, 1979.
\bibitem{muder} D.J.\ Muder, ``Minimal trellises for block codes,''
\emph{IEEE.\ Trans.\ Inform.\ Theory}, vol.\ 34, no.\ 5, pp.\ 1049--1053,
Sept.\ 1988.
\bibitem{oxley} J.G.\ Oxley, \emph{Matroid Theory}, Oxford University
Press, Oxford, UK, 1992.
\bibitem{RS-I} N.\ Robertson and P.D.\ Seymour, 
``Graph minors. I. Excluding a forest,''
\emph{J.\ Combin.\ Theory, Ser.\ B}, vol.\ 35, pp.\ 39--61, 1983.
\bibitem{vardy}
A.\ Vardy, ``Trellis Structure of Codes,'' 
in \emph{Handbook of Coding Theory}, R.\ Brualdi, C.\ Huffman and V.\ Pless,
Eds., Amsterdam, The Netherlands: Elsevier, 1998.
\end{thebibliography}
\end{document}